\documentclass[11pt,a4paper]{article}

\usepackage[utf8]{inputenc}
\usepackage[T1]{fontenc}
\usepackage[margin=2.5cm]{geometry}
\usepackage{amssymb,amsmath,amsthm,bm}
\usepackage{graphicx}
\usepackage[mathlines]{lineno}
\usepackage[authoryear,round]{natbib}
\usepackage[colorlinks=true,linkcolor=blue,citecolor=blue,urlcolor=blue,hyperfootnotes=false]{hyperref}

\hypersetup{
  pdftitle={Robust Model Order Selection via Dithered Differential Step-Down Thresholding},
  pdfauthor={Aleksandr Kharin}
}

\newtheorem{theorem}{Theorem}
\newtheorem{lemma}{Lemma}
\newtheorem{assumption}{Assumption}

\title{\textbf{Robust Model Order Selection via Dithered Differential Step-Down Thresholding}\\[0.5em]
\large \textit{Short title:} Dithered Differential Step-Down Thresholding}

\author{Aleksandr Kharin\textsuperscript{1,2,}\thanks{Corresponding author: Aleksandr Kharin, International Laboratory SSP-QF, Tomsk State University, 36 Lenin Ave., 634050 Tomsk, Russia. E-mail: \texttt{aleksandr.kharin.1989@gmail.com}}\\[0.6em]
\small \textsuperscript{1}International Laboratory SSP-QF, Tomsk State University, Tomsk, Russia\\[0.2em]
\small \textsuperscript{2}Laboratoire de Math\'ematiques Rapha\"el Salem, Universit\'e de Rouen-Normandie, Rouen, France}

\date{}

\begin{document}
\maketitle

\begin{abstract}
We propose a dithered differential step-down thresholding (DSDT) algorithm for robust model order selection in stationary noise. Evaluating only exceedance counts over a decreasing threshold grid, the proposed algorithm stops when detecting noise extrema value clustering. We prove that deterministic threshold grids are inconsistent under lowered floor detection thresholds due to threshold grid misalignment error. We bound this error using randomized grid dithering. For arbitrary deterministic shifts, any joint noise distribution, and data-dependent thresholds, we obtain non-asymptotic bounds on the estimation risk and family-wise error rate, establishing DSDT as a universal false alarm suppressor for any single-threshold algorithm. When the scale sequence tends to zero or extremes cluster along the search grid, DSDT achieves asymptotic exact order recovery of separated deterministic components exceeding the noise maximum plus the grid step, even below the minimax threshold or without a detection threshold. Using multiple independent dithers reduces error bounds. We provide convergence and consistency analysis for an arbitrary number of dithers. This analysis establishes asymptotic bounds under extreme value laws and non-asymptotic bounds under nondecreasing hazard rates. The results extend to unseparated deterministic shifts, random shifts, unknown noise scales, and multi-dimensional search grids. We show that DSDT is robust to data quantization.
\end{abstract}

\vspace{0.5em}
\noindent\textbf{Keywords:} extreme value theory; family-wise error rate; model order selection; randomized grid dithering; robust thresholding.

\medskip
\noindent\textbf{MSC 2020:} 60G70, 62G32, 62F05.

\section{Introduction}
In this paper, we study the problem of model order selection (MOS). Single-threshold detection at a minimax detection threshold is sensitive to extreme noise excursions \citep{Donoho1995, Piterbarg1996, Davies1987}. Moreover, in practice, lowering the detection threshold to capture weak deterministic components lying between the realized noise maximum and the minimax detection threshold introduces false alarms \citep{Abramovich2006, Massart2007}. Recent advances in MOS primarily focus on robust adaptive procedures and oracle inequalities to control such false alarms (e.g., \cite{Barber2015, Bogdan2015, Konev2012, Pchelintsev2018}).

We propose a dithered differential step-down thresholding (DSDT) algorithm to augment single-threshold detection. This algorithm evaluates a decreasing grid of thresholds down to the detection threshold of the original single-threshold algorithm (we call this threshold the {\it floor detection threshold}). The proposed algorithm uses the following stopping rule:
$\Delta\hat{\nu} \ge 2,$
where $\Delta\hat{\nu}$ is the difference between model order estimates for sequential thresholds. Note that if noise extrema exceed the floor detection threshold, this rule detects their value clustering and stops the DSDT algorithm, preventing their inclusion into the model. We prove that for any deterministic shifts and any floor detection threshold, the family-wise error rate of DSDT tends to zero. 

Under the assumed conditions, the deterministic shifts are separated, whereas the values of upper noise extrema cluster. We use extreme value theory \citep{Leadbetter1983, deHaan2006} to analyze the stationary noise. In this framework, the spacing between values of upper noise extrema asymptotically tends to zero in probability if at least one of the following two conditions holds. First, the scale sequence tends to zero \citep{Leadbetter1983, Embrechts1997} while the normalized spacing converges in distribution \citep{Leadbetter1983, David2003}. Second, correlated noise extrema cluster along the search grid \citep{Hsing1988, Aldous1989}.

Let us discuss related approaches. In time series analysis, \cite{Burridge2006} use extreme value theory to detect additive outliers. They apply a step-down test based on spacings between upper order statistics. In array processing, \cite{Kritchman2009} use a step-down test to find the number of signals. They compare sample eigenvalues to the Tracy-Widom distribution. In factor models, \cite{Onatski2010} estimates the number of factors as the last index at which the gap between consecutive ordered eigenvalues exceeds a threshold. In spiked covariance models, \citet{Passemier2012, Passemier2014} use the first index where single or consecutive eigenvalue spacings are below a threshold.
The proposed DSDT algorithm differs as follows. First, the stopping rule $\Delta\hat{\nu} \ge 2$ requires no quantiles and suppresses false alarms. Second, the algorithm uses only exceedance counts over a threshold grid without computing order statistics or their spacings. Third, using a discrete threshold grid introduces a threshold grid misalignment error, which we bound using randomized grid dithering \citep{Widrow2008}. We consider both a single dither and multiple independent dithers.

Let us list the contributions of this paper: (i) proving inconsistency for deterministic threshold grids; (ii) bounding the non-asymptotic family-wise error rate via randomized grid dithering; (iii) establishing asymptotic exact order recovery for arbitrarily lowered floor detection thresholds; (iv) extending results to arbitrary data-dependent thresholds; and (v) generalizing the results to multiple dithers and analyzing the convergence rate for an arbitrary number of dithers under extreme value laws and nondecreasing hazard rates.

\section{Problem Formulation and Algorithm}

\subsection{Observation Model and Risk}

In this paper, we estimate the true model order $\nu$ in stationary noise over a search grid of size $M$, where $\nu = o(M)$ as $M \to \infty$.
Suppose the observed data are mapped onto the search grid, yielding statistics $\mathcal{V} = \{V_1, \dots, V_M\}$:
\begin{equation} \label{SModel}
    V_i = d_i + \varepsilon_i, \quad i \in \{1, \dots, M\}.
\end{equation} 
Here $d_i \ge 0$ is a deterministic shift and $\{\varepsilon_i\}$ is a stationary noise sequence.  

\textbf{Remark 1.} In applications, the deterministic shifts $d_i > 0$ correspond to true signals (e.g., leading eigenvalues in multivariate statistics, discrete spectral lines in physical systems or non-null effect sizes in bioinformatics), whereas the sequence $\{\varepsilon_i\}$ represents the noise.

Next, we partition $\mathcal{V}$ into a deterministic subset $\mathcal{V}^S$ containing $\nu$ statistics with $d_i > 0$ (deterministic components), and a noise subset $\mathcal{V}^N$ containing $M - \nu$ statistics with $d_i = 0$ (noise components). Denote the decreasing sequence of order statistics from $\mathcal{V}^N$ by $V_{(1)}^N \ge V_{(2)}^N \ge \dots \ge V_{(M-\nu)}^N.$ Suppose that $M - \nu \ge 2$. We assume that the non-zero deterministic shifts scale proportionally to the typical maximum noise excursion (for instance, scaling as $\sqrt{\ln M}$ for Gaussian-type noise \citep{Donoho1995}).

Define the estimator at a threshold $\mathcal{T}$ as $\hat{\nu}(\mathcal{T}) = \sum_{i=1}^M \mathbf{1}_{\{V_i \ge \mathcal{T}\}}$. 
 The estimation risk $\mathbb{E}[L(\hat{\nu}, \nu)]$ for an estimator $\hat{\nu}$ is evaluated under the 0-1 loss $L(\hat{\nu}, \nu) = \mathbf{1}_{\{\hat{\nu} \neq \nu\}}$. An estimator $\hat{\nu}$ is consistent (achieves asymptotic exact order recovery) if \citep{Donoho1995,Ibragimov1981}
\[\mathbb{E}[L(\hat{\nu}, \nu)] = \Pr\left(\hat{\nu} \ne \nu \right) \xrightarrow[M \to \infty]{} 0.\]
Let us consider an estimator $\hat{\nu} = \hat{\nu}(\hat{\mathcal{T}})$. Let $N_{\mathrm{false}}$ be the number of noise components in the selected model $\{i:V_i \ge \hat{\mathcal{T}}\}$. We define a false alarm as the event $\{N_{\mathrm{false}} \ge 1\}$, and its probability is the family-wise error rate (FWER), $\mathrm{FWER}_{\hat{\nu}} = \Pr(N_{\mathrm{false}} \ge 1)$. Next, we define the false discovery rate (FDR) as $\mathrm{FDR}_{\hat{\nu}} = \mathbb{E}[N_{\mathrm{false}}/\max(\hat{\nu}, 1)]$. The overestimation event $\{\hat{\nu} > \nu\}$ implies $N_{\mathrm{false}} \ge 1$; thus, $\Pr(\hat{\nu} > \nu) \le \mathrm{FWER}_{\hat{\nu}}$. The event $\{N_{\mathrm{false}} \ge 1\}$ is equivalent to the event $\{V_{(1)}^N \ge \hat{\mathcal{T}}\}$. For the single-threshold estimator, $\mathrm{FWER}(\mathcal{T}) = \Pr(V_{(1)}^N \ge \mathcal{T})$.

\subsection{Algorithm} \label{algsec}
Firstly, let us introduce a floor detection threshold $\mathcal{T}_{\text{floor}}$. In this paper, we mostly consider $\mathcal{T}_{\text{floor}}$ as a minimax detection threshold $\mathcal{T}_{\text{mmx}}$ defined such that $\Pr(V_{(1)}^N \ge \mathcal{T}_{\text{mmx}}) \to 0$ as $M \to \infty$ (e.g., $\mathcal{T}_{\text{mmx}} \propto \sqrt{\ln M}$ for Gaussian-type tails \citep{Massart2007}); thus  $\text{FWER}(\mathcal{T}_{\text{mmx}}) \to 0$. Unless otherwise specified, $\mathcal{T}_{\text{floor}} = \mathcal{T}_{\text{mmx}}$. However, we also consider the case where $\mathcal{T}_{\text{floor}}$ is lowered below the minimax detection threshold. 

We apply randomized grid dithering to ensure statistical uniformity over the threshold grid \citep{Gray1993, Widrow2008}. For a threshold step size $\Delta\mathcal{T} > 0$, we create a random variable (dither) $\xi \sim \text{Unif}(0, \Delta\mathcal{T})$ independently of the observed data. 
The threshold sequence is \[\mathcal{T}_k = \mathcal{T}_{\text{max}} - k\Delta\mathcal{T} + \xi\] for all integers $k \ge 0$. Note that $\mathcal{T}_{\text{max}}$ is chosen \textit{a priori} independent of the observed data (unless otherwise specified, see Corollary~3) and such that $\mathcal{T}_{\text{max}} \ge \mathcal{T}_{\text{floor}}$ and $\Pr(V_{(1)}^N \ge \mathcal{T}_{\text{max}}) \to 0$ (e.g., $C \sqrt{\ln M}$ for Gaussian-type noise). 

At each step $k \ge 1$, the proposed DSDT algorithm computes $\Delta\hat{\nu}_k = \hat{\nu}(\mathcal{T}_k) - \hat{\nu}(\mathcal{T}_{k-1})$. The stopping index is
\[k_{\text{stop}} =  \min \big\{ k \ge 1 : \Delta\hat{\nu}_k \ge 2 \text{ or } \mathcal{T}_k \le \mathcal{T}_{\text{floor}} \text{ or } \hat{\nu}(\mathcal{T}_k) = M \big\}. \]
The model order estimate is $\hat{\nu}_{\text{DSDT}} = \hat{\nu}(\mathcal{T}_{k_{\text{stop}}-1})$.

\section{Assumptions and preliminary results}
\subsection{Deterministic shifts}
 Let $\{V_i^S\}_{i=1}^\nu$ be the elements of the deterministic subset $\mathcal{V}^S$, and let $\{d_i^S\}_{i=1}^\nu$ be their corresponding non-zero deterministic shifts. Without loss of generality, we assume that the non-zero deterministic shifts satisfy $d_{1}^S \ge d_{2}^S \ge \dots \ge d_{\nu}^S$.

\begin{assumption} \label{defgap}
 Assume that $\{d_{i}^S\}_{i=1}^{\nu}$ satisfy:
 
 (a) $d_{\nu}^S \ge \mathcal{T}_{\text{mmx}} + C_D(M)$, ${\varepsilon_{\max}^S = o_{\Pr}(C_D(M))}$;
 
 (b) $\min_{1 \le i < \nu} (d_{i}^S - d_{i+1}^S) \ge C_S(M)$, ${\varepsilon_{\max}^S = o_{\Pr}(C_S(M))}$. 
 
 Here $C_D(M)$ and $C_S(M)$ are sequences satisfying $\lim_{M \to \infty} C_D(M) = \infty$ and $\lim_{M \to \infty} C_S(M) = \infty$, $\varepsilon_i^S = V_i^S - d_i^S$, and $\varepsilon_{\max}^S = \max_{1 \le i \le \nu} |\varepsilon_i^S|$.
\end{assumption}

The conditions on $C_D(M)$ and $C_S(M)$ establish sufficient conditions for asymptotic exact order recovery. Assuming $\mathcal{T}_{\text{floor}} = \mathcal{T}_{\text{mmx}}$, they ensure that both the gaps between consecutive deterministic components and the distance from the weakest deterministic component to the floor detection threshold asymptotically exceed the threshold step size $\Delta\mathcal{T}$. In finite samples, these gaps must additionally exceed $\Delta\mathcal{T}$ by a margin that covers the noise.

\textbf{Remark 2.} Assumption \ref{defgap}(b) is needed because the stopping rule $\Delta\hat{\nu}_k \ge 2$ cannot resolve deterministic components with identical shifts ($d_{i}^S = d_{i+1}^S$). Separated deterministic shifts arise in various fields: (1) \textit{Multivariate statistics:} In principal component analysis, distinct latent factors yield non-degenerate eigenvalues. (2) \textit{Physical spectral analysis:} In dissipative oscillatory systems, the amplitudes of harmonic overtones decay monotonically. (3) \textit{Bioinformatics:} In differential gene expression analysis, biological heterogeneity ensures that non-null effect sizes of distinct regulatory genes do not perfectly coincide.  

\begin{lemma} \label{concentr}
Define the events
\[\Omega_D = \big\{ \min_{1 \le i \le \nu} V_{i}^S > \max(V_{(1)}^N, \mathcal{T}_{\text{floor}}) + \Delta\mathcal{T},\ V_{(1)}^N < \mathcal{T}_{\text{max}} \big\},\]
\[\Omega_S = \big\{ \min_{1 \le i < \nu} (V_{i}^S - V_{i+1}^S) > \Delta\mathcal{T} \big\}.\]
If Assumption \ref{defgap}(b) holds, then $\Pr(\Omega_S) \to 1$ as $M \to \infty$. If Assumption \ref{defgap}(a) holds and $\mathcal{T}_{\text{floor}} = \mathcal{T}_{\text{mmx}}$, then $\Pr(\Omega_D) \to 1$ as $M \to \infty$.
\end{lemma}
\begin{proof}
By Assumption~\ref{defgap}(b),
$ \min_{1 \le i < \nu} (V_i^S - V_{i+1}^S) \ge \min_{1 \le i < \nu} (d_i^S - d_{i+1}^S) - 2\varepsilon_{\max}^S \ge C_S(M) - 2\varepsilon_{\max}^S \xrightarrow{\Pr} \infty. $
Since $\Delta\mathcal{T}$ is bounded, this yields $\Pr(\Omega_S) \to 1$.

For $\mathcal{T}_{\text{floor}} = \mathcal{T}_{\text{mmx}}$, $\Pr(V_{(1)}^N < \mathcal{T}_{\text{floor}}) \to 1$ by the definition of $\mathcal{T}_{\text{mmx}}$. On $\{V_{(1)}^N < \mathcal{T}_{\text{floor}}\} \subset \{V_{(1)}^N < \mathcal{T}_{\text{max}}\}$, Assumption~\ref{defgap}(a) implies
$ \min_{1 \le i \le \nu} V_i^S - \max(V_{(1)}^N, \mathcal{T}_{\text{floor}}) \ge d_\nu^S - \varepsilon_{\max}^S - \mathcal{T}_{\text{floor}} \ge C_D(M) - \varepsilon_{\max}^S \xrightarrow{\Pr} \infty, $
thus $\Pr(\Omega_D) \to 1$.
\end{proof}

\subsection{Noise assumptions} \label{extras}
Denote $G = V_{(1)}^N - V_{(2)}^N$. Let $n = M - \nu$.

\begin{assumption} \label{noiseass}
(a) The sequence $\{\varepsilon_i\}$ satisfies Leadbetter's $D(u_n)$ mixing condition and possesses an extremal index $\theta \in (0, 1]$ \citep{Leadbetter1983}. Its distribution belongs to a maximum domain of attraction with norming sequences $a_n > 0$ and $b_n$ such that $a_n^{-1}(V_{(1)}^N - b_n)$ converges in distribution \citep{deHaan2006}, where $D(u_n)$ holds for every $u_n = a_n x + b_n$ with fixed $x \in \mathbb{R}$. Next, assume that at least one of the following conditions holds:
(i) The scale sequence satisfies $a_n \to 0$.
(ii) Extremes cluster along the search grid (i.e., $\theta < 1$) and $V_{(1)}^N - V_{\mathcal{C}_1}^{(2)} \xrightarrow{\Pr} 0$, where $V_{\mathcal{C}_1}^{(2)}$ is the second-largest value in the search grid cluster $\mathcal{C}_1$ containing $V_{(1)}^N$.
(b) The sequence $\mathcal{T}_{\text{max}}(M)$ is deterministic and satisfies $\Pr(V_{(1)}^N \ge \mathcal{T}_{\text{max}}) \to 0$.
\end{assumption}

\begin{lemma} \label{suff}
Under Assumption \ref{noiseass}(a), we get: 
$G \xrightarrow[M \to \infty]{\Pr} 0$.
\end{lemma}
\begin{proof}
Let us analyze the two alternative conditions of Assumption \ref{noiseass}(a).

Suppose condition (i) holds. Let us partition the $n$ noise components into two disjoint consecutive blocks of sizes $n_1 = \lfloor n/2 \rfloor$ and $n_2 = n - n_1$. Let $Z_1$ and $Z_2$ be maxima of these blocks, then $V_{(1)}^N = \max(Z_1, Z_2)$ and $V_{(2)}^N \ge \min(Z_1, Z_2)$. Therefore, $0 \le G \le |Z_1 - Z_2|$.
Under our assumptions, the normalized block maxima $a_n^{-1}(Z_1 - b_n)$ and $a_n^{-1}(Z_2 - b_n)$  converge in distribution \citep[Lemma 3.3.2, Theorem 3.3.3]{Leadbetter1983}, \citep{deHaan2006}. This implies $a_n^{-1}|Z_1 - Z_2| = O_{\Pr}(1)$. Because $a_n \to 0$, we conclude $|Z_1 - Z_2| \xrightarrow{\Pr} 0$, yielding $G \xrightarrow{\Pr} 0$.

Alternatively, suppose condition (ii) holds. By definition, $V_{(2)}^N \ge V_{\mathcal{C}_1}^{(2)}$. Thus:
$0 \le G = V_{(1)}^N - V_{(2)}^N \le V_{(1)}^N - V_{\mathcal{C}_1}^{(2)}$.
Using condition (ii), we obtain $V_{(1)}^N - V_{\mathcal{C}_1}^{(2)} \xrightarrow{\Pr} 0$, yielding $G \xrightarrow{\Pr} 0$.
\end{proof}
\subsection{Additional conditions}
\begin{assumption} \label{sep_assum}
Assume that $\mathcal{T}_{\text{floor}}$ is lowered such that $\lim_{M \to \infty} \Pr(V_{(1)}^N \ge \mathcal{T}_{\text{floor}}) = 1$, while
\begin{equation} \label{condtn}
 \Pr\big(\min_{1 \le i \le \nu} V_i^S > \max(V_{(1)}^N, \mathcal{T}_{\text{floor}}) + \Delta\mathcal{T}\big) \xrightarrow[M \to \infty]{} 1.
\end{equation}
\end{assumption}
\begin{assumption} \label{detga}
For a deterministic threshold grid ($\xi \equiv 0$), the sequence $\mathcal{T}_{\text{max}}(M)$ satisfies $\lim_{M \to \infty} (b_n - \mathcal{T}_{\text{max}}) = -\infty$, $(b_n - \mathcal{T}_{\text{max}})/a_n \to -\infty$, and $\mathcal{T}_{\text{max}}(M+1) - \mathcal{T}_{\text{max}}(M) = o(a_n)$.
\end{assumption}

\section{Consistency Analysis} \label{secconsist}

\begin{lemma} \label{lem_dithering}
Let $\xi \sim \mathrm{Unif}(0, \Delta\mathcal{T})$ be independent of a random vector $\mathcal{V}$. For any measurable function $g$, let $U = (g(\mathcal{V}) - \xi) \bmod \Delta\mathcal{T}$. Then $U \sim \mathrm{Unif}(0, \Delta\mathcal{T})$ and $U$ is independent of $\mathcal{V}$.
\end{lemma}
\begin{proof}
For any $y \in \mathbb{R}$, the mapping $s \mapsto (y - s) \bmod \Delta\mathcal{T}$ preserves the Lebesgue measure on $[0, \Delta\mathcal{T})$. Since $\xi$ is independent of $\mathcal{V}$,  $\Pr(U \le u \mid \mathcal{V}) = u/\Delta\mathcal{T}$ for $u \! \in \! [0, \Delta\mathcal{T}]$, thus $U \! \sim \! \mathrm{Unif}(0, \Delta\mathcal{T})$ and $U$ is independent of $\mathcal{V}$.
\end{proof}

\begin{lemma} \label{grid_misalignment}
Let $k$ be such that $V_{(1)}^N \in [\mathcal{T}_k, \mathcal{T}_{k-1})$. Define the threshold grid misalignment event as $\mathcal{E}_{\text{mis}} = \{ \mathcal{T}_k \in (V_{(2)}^N, V_{(1)}^N] \}$. Then $\Pr(\mathcal{E}_{\text{mis}}) = \mathbb{E}[\min(G, \Delta\mathcal{T})]/\Delta\mathcal{T}$.
\end{lemma}
\begin{proof}
The event $\mathcal{E}_{\text{mis}}$ occurs if and only if $U = V_{(1)}^N - \mathcal{T}_k < G$. Taking into account that $U = (V_{(1)}^N - \mathcal{T}_{\text{max}} - \xi) \bmod \Delta\mathcal{T}$ and using Lemma \ref{lem_dithering}, one can conclude $U \sim \text{Unif}(0, \Delta\mathcal{T})$ and $U$ is independent of the observed data. 
Since $U \sim \text{Unif}(0, \Delta\mathcal{T})$ and is independent of $G$, we get:
$\Pr(\mathcal{E}_{\text{mis}}) = \mathbb{E}[\Pr(G > U \mid G)] = \frac{\mathbb{E}[\min(G, \Delta\mathcal{T})]}{\Delta\mathcal{T}}.$
\end{proof}

\begin{lemma} \label{frisk}
Let
$W =\min\big(G,\ (V_{(1)}^N-\mathcal{T}_{\text{floor}})_+,\ \Delta\mathcal{T}\big) ,$
where $(x)_+ = \max(x, 0)$ and $\mathcal{T}_{\text{floor}}$ does not necessarily equal the minimax detection threshold.
The following inequality holds:
\begin{equation*}
|\mathbb{E}[L(\hat{\nu}_{\text{DSDT}}, \nu)] - {\mathbb{E}[W]}/{\Delta\mathcal{T}}| \le \Pr(\Omega_S^c) + \Pr(\Omega_D^c).
\end{equation*}
\end{lemma}
\begin{proof}
Let us consider the event ${\Omega_S \cap \Omega_D}$. Since $\Omega_D \subset \{V_{(1)}^N < \mathcal{T}_{\text{max}}\}$, there is $k \ge 1$ such that $V_{(1)}^N \in [\mathcal{T}_k, \mathcal{T}_{k-1})$.
For any $j < k$, $\Omega_S$ yields $\Delta\hat{\nu}_j \le 1$, thus the stopping rule $\Delta\hat{\nu} \ge 2$ is not met.

If $V_{(1)}^N < \mathcal{T}_{\text{floor}}$, the DSDT algorithm stops at $k_{\text{stop}} \le k$ satisfying $\mathcal{T}_{k_{\text{stop}}} \le \mathcal{T}_{\text{floor}}$. By $\Omega_D$, all deterministic components exceed $\mathcal{T}_{\text{floor}} + \Delta\mathcal{T} \ge \mathcal{T}_{k_{\text{stop}}-1}$. Since $\mathcal{T}_{k_{\text{stop}}-1} \ge \mathcal{T}_{k-1} > V_{(1)}^N$, we get $\hat{\nu}(\mathcal{T}_{k_{\text{stop}}-1}) = \nu$. 

Next, if $V_{(1)}^N \ge \mathcal{T}_{\text{floor}}$, then for any $j < k$, $\mathcal{T}_j \ge \mathcal{T}_{k-1} > V_{(1)}^N \ge \mathcal{T}_{\text{floor}}$, thus the DSDT algorithm reaches step $k$. By $\Omega_D$, all deterministic components exceed $V_{(1)}^N + \Delta\mathcal{T} \ge \mathcal{T}_{k-1}$, yielding $\hat{\nu}(\mathcal{T}_{k-1}) = \nu$. In this case, there are two disjoint events. If $ V_{(2)}^N \ge \mathcal{T}_k $, then $\hat{\nu}_{\text{DSDT}} = \nu$.
If $ V_{(2)}^N < \mathcal{T}_k$, then $\Delta\hat{\nu}_k = 1$. Thus the DSDT algorithm stops if and only if $\mathcal{T}_k \le \mathcal{T}_{\text{floor}}$, yielding $\hat{\nu}_{\text{DSDT}} = \nu$. If $ \mathcal{T}_k > \mathcal{T}_{\text{floor}}$, the algorithm evaluates lower thresholds, yielding $\hat{\nu}_{\text{DSDT}} \ge \nu + 1$.
Therefore, on the event $\Omega_S \cap \Omega_D$, an estimation error occurs if and only if $V_{(2)}^N < \mathcal{T}_k$ and $\mathcal{T}_k > \mathcal{T}_{\text{floor}}$. With $U = V_{(1)}^N - \mathcal{T}_k$, using Lemma \ref{lem_dithering} as in the proof of Lemma \ref{grid_misalignment}, $U \sim \text{Unif}(0, \Delta\mathcal{T})$ and is independent of the observed data.
The error conditions are $\left\lbrace U < G\right\rbrace \cap \left\lbrace U < V_{(1)}^N - \mathcal{T}_{\text{floor}}\right\rbrace $. Since ${U \ge 0}$, an error occurs on the event $\Omega_S \cap \Omega_D$ if and only if $U < W$. Thus,
$\{\hat{\nu}_{\text{DSDT}} \neq \nu\} \cap \Omega_S \cap \Omega_D = \{U < W\} \cap \Omega_S \cap \Omega_D.$
It implies:
$ \big| \Pr(\{\hat{\nu}_{\text{DSDT}} \neq \nu\}) - \Pr(U < W) \big| \le \Pr(\Omega_S^c) + \Pr(\Omega_D^c). $
Finally, since $U$ is independent of $W$, $\Pr(U < W) = \mathbb{E}[W]/\Delta\mathcal{T}$.
\end{proof}

Now let us demonstrate the limitations of the algorithm with deterministic threshold grid and lowered floor detection threshold with the following result.
\begin{theorem} \label{detfail}
Suppose Assumptions \ref{defgap}(b), \ref{noiseass}(a), \ref{sep_assum}, and \ref{detga} hold, the threshold grid is deterministic ($\xi \equiv 0$), the noise distribution belongs to the Gumbel maximum domain of attraction, and $\theta = 1$. Then, differential step-down thresholding is inconsistent, yielding $\limsup_{M \to \infty} \mathbb{E}[L] \ge e^{-1}$.
\end{theorem}
\begin{proof}
Let $\mathcal{T}_k = \mathcal{T}_{\text{max}} - k\Delta\mathcal{T}$. On $\{V_{(1)}^N < \mathcal{T}_{\text{max}}\}$, the misalignment event $\mathcal{E}_{\text{mis}}$ occurs if and only if 
\begin{equation} \label{condth1}
{G > (V_{(1)}^N - \mathcal{T}_{\text{max}}) \bmod{\Delta\mathcal{T}}}.
\end{equation}
Using \cite[Theorem 1.2.3]{Leadbetter1983}, one can obtain $b_{n+1} - b_n = o(a_n)$. Under Assumption \ref{detga}, the sequence $b_n - \mathcal{T}_{\text{max}} \to -\infty$ with increments $o(a_n)$. This implies the existence of a subsequence $\{M_j\}$ (with $n_j = M_j - \nu$) where the remainder $\delta_j = (b_{n_j} - \mathcal{T}_{\text{max}}(M_j)) \bmod{\Delta\mathcal{T}} \in [0, \Delta\mathcal{T})$ satisfies $\delta_j = o(a_{n_j})$. 

Using condition (i) of Assumption \ref{noiseass}(a) and Lemma \ref{suff}, we get $G \xrightarrow{\Pr} 0$. Next, Assumption \ref{detga} implies Assumption \ref{noiseass}(b).
Let $u_{n_j} = b_{n_j} - \delta_j$ and $X_j = V_{(1)}^N - u_{n_j}$. Now the condition \eqref{condth1} along $\{M_j\}$ can be rewritten as ${G > X_j \bmod \Delta\mathcal{T}}$. 
Since $V_{(1)}^N - b_{n_j} = O_{\Pr}(a_{n_j})$, $\delta_j = o(a_{n_j})$, $a_{n_j} \to 0$, we get $X_j = V_{(1)}^N - b_{n_j} + \delta_j \xrightarrow{\Pr} 0$.

Let us analyze two disjoint events: $\{X_j \ge 0\}$ and ${\{X_j < 0\}}$.
If $X_j \ge 0$, since $X_j \xrightarrow{\Pr} 0$, the event $\{0 \le X_j < \Delta\mathcal{T}\}$ is asymptotically equivalent to $\{X_j \ge 0\}$. Thus, 
$\{X_j \ge 0, \ G > X_j\}=\{V_{(1)}^N \ge u_{n_j} > V_{(2)}^N\}.$ 
For fixed $x \in \mathbb{R}$ and $u_n(x) = a_n x + b_n$, partition $\mathcal{V}^N$ into $k$ consecutive blocks of sizes $\lfloor n/k \rfloor$ or $\lceil n/k \rceil$ with maxima $Z_l$ and exceedance counts $N_l$ above $u_n(x)$ ($l \le k$). Under condition $D(u_n(x))$ and $\theta = 1$, as $n \to \infty$, $\mathbb{E}[N_l] \to e^{-x}/k$ and $\Pr(N_l \ge 1) = \Pr(Z_l > u_n(x)) \to 1 - \exp(-e^{-x}/k)$ \citep{Leadbetter1983}; thus $\Pr(N_l \ge 2) \le \mathbb{E}[N_l] - \Pr(N_l \ge 1)$ yields $\limsup_{n\to\infty} \sum_{l=1}^k \Pr(N_l \ge 2) \le e^{-x} - k(1 - \exp(-e^{-x}/k)) \to 0$ as $k \to \infty$. Suppose $B_k$ is the event that at most one $Z_l$ exceeds $u_n(x)$, then $0 \le \Pr(B_k) - \Pr(V_{(2)}^N \le u_n(x)) \le \sum_{l=1}^k \Pr(N_l \ge 2)$. Since $\Pr(B_k) = \sum_{l=1}^k \Pr(\max_{m \ne l} Z_m \le u_n(x)) - (k-1)\Pr(\max_{m \le k} Z_m \le u_n(x))$, $D(u_n(x))$ \citep[Lemmas~3.2.2 and~3.3.2]{Leadbetter1983} yields $\lim_{n\to\infty} \Pr(B_k) = k \exp(-e^{-x}(k-1)/k) - (k-1)\exp(-e^{-x}) \to \exp(-e^{-x})(1+e^{-x})$ as $k \to \infty$, this implies $\Pr(V_{(2)}^N \le u_n(x)) \to \exp(-e^{-x})(1+e^{-x})$. Since $\delta_j = o(a_{n_j})$, for any $\varepsilon > 0$ and large $j$, $a_{n_j}(-\varepsilon) + b_{n_j} < u_{n_j} < a_{n_j}\varepsilon + b_{n_j}$. Next, we get $\Pr(V_{(1)}^N < u_{n_j}) \to e^{-1}$ and $\Pr(V_{(2)}^N < u_{n_j}) \to 2e^{-1}$. Finally, ${\Pr(V_{(1)}^N \ge u_{n_j} > V_{(2)}^N) = \Pr(V_{(2)}^N < u_{n_j}) - \Pr(V_{(1)}^N < u_{n_j}) \to e^{-1}}$.

In the case $X_j < 0$, $X_j \bmod \Delta\mathcal{T} = X_j + \Delta\mathcal{T}$ (since ${\Pr(X_j \ge -\Delta\mathcal{T}) \to 1}$). Since $X_j \xrightarrow{\Pr} 0$ and $G \xrightarrow{\Pr} 0$, we get $\Pr(\{X_j < 0, \ G > X_j + \Delta\mathcal{T}\}) \to 0$.

The overall limit along $\{M_j\}$ is $\lim_{j \to \infty} \Pr(\mathcal{E}_{\text{mis}}) = e^{-1}$.

In the same way as in Lemma \ref{frisk} (for $\xi \equiv 0$), one can obtain:
$ \mathbb{E}[L] \ge \Pr(\mathcal{E}_{\text{mis}} \cap \{V_{(2)}^N \ge \mathcal{T}_{\text{floor}}\} \cap \Omega_S \cap \Omega_D). $
Assumption \ref{sep_assum} implies $a_n^{-1}(\mathcal{T}_{\text{floor}} - b_n) \to -\infty$, ensuring $\Pr(V_{(2)}^N < \mathcal{T}_{\text{floor}}) \to 0$. Furthermore, condition \eqref{condtn} and Lemma \ref{concentr} imply $\mathbb{E}[L] \ge \Pr(\mathcal{E}_{\text{mis}}) - o(1)$. Thus, along $\{M_j\}$, $\limsup_{M \to \infty} \mathbb{E}[L] \ge e^{-1}$.
\end{proof}

\begin{theorem} \label{main_fin}
For arbitrary deterministic shifts $d_i \ge 0$, for any floor detection threshold $\mathcal{T}_{\mathrm{floor}} \le \mathcal{T}_{\mathrm{max}}$, and for any joint distribution of the noise components, the following inequalities hold:
$\mathrm{FWER}_{\mathrm{DSDT}} \le \mathrm{FWER}(\mathcal{T}_{\mathrm{floor}})$ and
\begin{equation} \label{eq_fwer_bound}
\mathrm{FWER}_{\mathrm{DSDT}} \le {\mathbb{E}[W]}/{\Delta\mathcal{T}} + \Pr(V_{(1)}^N \!\ge \mathcal{T}_{\mathrm{max}}),
\end{equation}
where $W =  \min\!\big(G,(V_{(1)}^N-\mathcal{T}_{\mathrm{floor}})_+,\Delta\mathcal{T}\big)$.
\end{theorem}

\begin{proof}
Firstly, $\mathrm{FWER}_{\mathrm{DSDT}} = \Pr(V_{(1)}^N \ge \mathcal{T}_{k_{\mathrm{stop}}-1}) \le \Pr(V_{(1)}^N \ge \mathcal{T}_{\mathrm{floor}}) = \mathrm{FWER}(\mathcal{T}_{\mathrm{floor}})$.

On $\{V_{(1)}^N < \mathcal{T}_{\mathrm{max}}\}$, let $k \ge 1$ be such that $V_{(1)}^N \in [\mathcal{T}_k, \mathcal{T}_{k-1})$. Next, $V_{(1)}^N \ge \mathcal{T}_{k_{\mathrm{stop}}-1}$ implies $k_{\mathrm{stop}} > k$, $\Delta\hat{\nu}_k \le 1$. Thus $V_{(2)}^N < \mathcal{T}_k$ and $\mathcal{T}_k > \mathcal{T}_{\mathrm{floor}}$. 
Let $U = V_{(1)}^N - \mathcal{T}_k = (V_{(1)}^N - \mathcal{T}_{\mathrm{max}} - \xi) \bmod \Delta\mathcal{T}$. From Lemma~\ref{lem_dithering}, we conclude that $U \sim \mathrm{Unif}(0, \Delta\mathcal{T})$ and is independent of $W$. Let us rewrite $V_{(2)}^N < \mathcal{T}_k$ and $\mathcal{T}_k > \mathcal{T}_{\mathrm{floor}}$ as $U < G$ and $U < V_{(1)}^N - \mathcal{T}_{\mathrm{floor}}$, respectively. Thus, on the event $\{V_{(1)}^N < \mathcal{T}_{\mathrm{max}}\}$, we get ${\{V_{(1)}^N \ge \mathcal{T}_{k_{\mathrm{stop}}-1}\} \subseteq \{U < W\}}$. Finally, $\mathrm{FWER}_{\mathrm{DSDT}} \le \Pr(U < W) + \Pr(V_{(1)}^N \ge \mathcal{T}_{\mathrm{max}})$ and ${\Pr(U < W) = \mathbb{E}[W]/\Delta\mathcal{T}}$.
\end{proof}

\begin{theorem} \label{th_ens}
Suppose that Assumption \ref{defgap} holds, and $\mathcal{T}_{\text{floor}} = \mathcal{T}_{\text{mmx}}$. In this case, the proposed DSDT algorithm achieves asymptotic exact order recovery: $\mathbb{E}[L(\hat{\nu}_{\text{DSDT}}, \nu)] \xrightarrow[M \to \infty]{} 0.$
\end{theorem}
\begin{proof}
Using Lemma \ref{concentr}, we get $\Pr(\Omega_S \cap \Omega_D) \to 1$. Since $\mathcal{T}_{\text{floor}} = \mathcal{T}_{\text{mmx}}$, $(V_{(1)}^N - \mathcal{T}_{\text{floor}})_+ \xrightarrow{\Pr} 0$. Since the random variable $\min(G, (V_{(1)}^N - \mathcal{T}_{\text{floor}})_+, \Delta\mathcal{T})$ is bounded by $\Delta\mathcal{T}$, using Bounded Convergence Theorem (BCT) and applying Lemma~\ref{frisk}, we obtain $\mathbb{E}[L(\hat{\nu}_{\text{DSDT}}, \nu)] \xrightarrow[M \to \infty]{} 0$.
\end{proof}

\begin{theorem} \label{main_inf}
Suppose that Assumption \ref{noiseass} holds and that the floor detection threshold $\mathcal{T}_{\text{floor}}$ may be arbitrarily lowered below the minimax detection threshold even if $\text{FWER}(\mathcal{T}_{\text{floor}}) \to 1$.

(a) If Assumption \ref{defgap}(b) and condition \eqref{condtn} hold, then the proposed DSDT algorithm achieves asymptotic exact order recovery: $\mathbb{E}[L(\hat{\nu}_{\text{DSDT}}, \nu)] \xrightarrow[M \to \infty]{} 0$.

(b) For arbitrary deterministic shifts,
 $\text{FWER}_{\text{DSDT}} \xrightarrow[M \to \infty]{} 0.$
\end{theorem}
\begin{proof}
(a) Using Lemma \ref{concentr}, we get $\Pr(\Omega_S^c) \to 0$. Next, using condition \eqref{condtn} with Assumption \ref{noiseass}(b), we obtain $\Pr(\Omega_D^c) \to 0$. By Lemma \ref{suff}, $G \xrightarrow{\Pr} 0$, thus $\min(G, (V_{(1)}^N - \mathcal{T}_{\text{floor}})_+, \Delta\mathcal{T}) \xrightarrow{\Pr} 0$. Using BCT, we get
$\mathbb{E}[\min(G, (V_{(1)}^N - \mathcal{T}_{\text{floor}})_+, \Delta\mathcal{T})] \xrightarrow[M \to \infty]{} 0.$ 
Next, applying Lemma \ref{frisk}, we obtain $\mathbb{E}[L(\hat{\nu}_{\text{DSDT}}, \nu)] \xrightarrow[M \to \infty]{} 0$.

(b) By Theorem \ref{main_fin}, $\text{FWER}_{\text{DSDT}} \le \mathbb{E}[\min(G, \Delta\mathcal{T})]/\Delta\mathcal{T} + \Pr(V_{(1)}^N \ge \mathcal{T}_{\text{max}})$. By Lemma~\ref{suff} and BCT, $\mathbb{E}[\min(G, \Delta\mathcal{T})] \to 0$. By Assumption \ref{noiseass}(b), ${\Pr(V_{(1)}^N \ge \mathcal{T}_{\text{max}}) \to 0}$. 
\end{proof}

\textbf{Corollary 1.} Under the conditions of Theorem \ref{main_inf}(a), asymptotic exact order recovery is preserved if $\mathcal{T}_{\text{floor}} = -\infty$. It means that under conditions of Theorem \ref{main_inf}(a), the stopping rule $\Delta\hat{\nu} \ge 2$ requires no detection threshold.

\textbf{Corollary 2.} Suppose the conditions of Theorem \ref{main_inf}(a) hold, with the exception that Assumption \ref{noiseass}(a) is not necessarily satisfied. If $\Pr(V_{(1)}^N - \mathcal{T}_{\text{floor}} \ge \Delta\mathcal{T}) \to 1$, the proposed algorithm achieves asymptotic exact order recovery if and only if $G \xrightarrow{\Pr} 0$.
\begin{proof}
By the stated conditions, $\Pr(\Omega_S \cap \Omega_D) \to 1$. Since $\Pr(V_{(1)}^N - \mathcal{T}_{\text{floor}} \ge \Delta\mathcal{T}) \to 1$, $W$ in Lemma \ref{frisk} satisfies $\mathbb{E}[W] - \mathbb{E}[\min(G, \Delta\mathcal{T})] \to 0$ by the BCT. Thus, Lemma \ref{frisk} yields $\mathbb{E}[L(\hat{\nu}_{\text{DSDT}}, \nu)] - \Delta\mathcal{T}^{-1} \mathbb{E}[\min(G, \Delta\mathcal{T})] \to 0$ as $M \to \infty$. For a fixed $\Delta\mathcal{T} > 0$, since $0 \le \min(G, \Delta\mathcal{T}) \le \Delta\mathcal{T}$, we have $\mathbb{E}[\min(G, \Delta\mathcal{T})] \to 0$ if and only if $G \xrightarrow{\Pr} 0$.
\end{proof}

\textbf{Corollary 3.} Setting $\mathcal{T}_{\text{max}} = \max(\max_{1 \le i \le M} V_i, \mathcal{T}_{\text{floor}}) + \delta$ for any fixed $\delta > 0$ yields $\Pr(V_{(1)}^N \ge \mathcal{T}_{\text{max}}) = 0$, eliminating Assumption \ref{noiseass}(b). Applying Lemma \ref{lem_dithering} to $g(\mathcal{V}) = V_{(1)}^N - \mathcal{T}_{\text{max}}$ ensures $U \sim \text{Unif}(0, \Delta\mathcal{T})$ remains independent of $\mathcal{V}$, preserving Theorems \ref{main_fin} and \ref{main_inf}.

\textbf{Corollary 4.} For any data-dependent threshold $\mathcal{T}_{\text{floor}}(\mathcal{V}) \le \mathcal{T}_{\text{max}}$, Lemma \ref{lem_dithering} ensures $U$ remains independent of $\mathcal{V}$, preserving Theorem \ref{main_fin} and $\text{FWER}_{\text{DSDT}} \to 0$ in Theorem \ref{main_inf}(b). Moreover, by the proof of Lemma \ref{frisk}, on $\Omega_S \cap \Omega_D$, DSDT errs only if $\hat{\nu}(\mathcal{T}_{\text{floor}})$ errs, and then with conditional probability $W/\Delta\mathcal{T}$ given $\mathcal{V}$; hence, $\mathbb{E}[L(\hat{\nu}_{\text{DSDT}}, \nu)] \le \mathbb{E}[L(\hat{\nu}(\mathcal{T}_{\text{floor}}), \nu)\, W]/\Delta\mathcal{T} + \Pr(\Omega_S^c) + \Pr(\Omega_D^c)$. Thus, DSDT can be used as a false alarm suppressor for any single-threshold algorithm.

\section{Convergence analysis}
In the previous sections, we analyzed the DSDT algorithm with a single dither. In this section we generalize our approach to multiple dithers to improve the convergence rate. Next, we provide a convergence analysis for an arbitrary number of dithers. 

Let $\hat{\nu}^{(j)}$ be the DSDT estimator with a dither $\xi_j$, where $\left\lbrace \xi_j\right\rbrace_{j=1}^K$ are i.i.d.\ $\mathrm{Unif}(0,\Delta\mathcal{T})$ and independent of $\mathcal{V}$, and let $\hat{\nu}_K = \min_{j \le K} \hat{\nu}^{(j)}$ (if $K = 1$, we get DSDT). Note that, the consistency results of Section \ref{secconsist} extend to any $K \ge 1$ via the following generalization of Lemma~\ref{frisk} and Theorem~\ref{main_fin}.

\begin{theorem} \label{lemrate}
For any $K \ge 1$, any deterministic shifts, any joint distribution of the noise and any thresholds $\mathcal{T}_{\mathrm{floor}} \le \mathcal{T}_{\mathrm{max}}$,
\begin{equation*}
\begin{split}
\mathrm{FWER}_{\hat{\nu}_K} \le \mathbb{E}[(W/\Delta\mathcal{T})^K] + \Pr(V_{(1)}^N \ge \mathcal{T}_{\mathrm{max}}), \quad
|\mathbb{E}[L(\hat{\nu}_K, \nu)] - \mathbb{E}[(W/\Delta\mathcal{T})^K]| \le \Pr(\Omega_S^c \cup \Omega_D^c),\\
0 \le \mathbb{E}[(\min(G,\Delta\mathcal{T})/\Delta\mathcal{T})^K] - \mathbb{E}[(W/\Delta\mathcal{T})^K] \le \Pr(V_{(1)}^N < \mathcal{T}_{\mathrm{floor}} + \Delta\mathcal{T}).
\end{split}
\end{equation*}
\end{theorem}
\begin{proof}
From Lemma~\ref{lem_dithering}, $U_j = (V_{(1)}^N - \mathcal{T}_{\text{max}} - \xi_j) \bmod \Delta\mathcal{T}$, $j \le K$, are i.i.d.\ $\mathrm{Unif}(0,\Delta\mathcal{T})$ given $\mathcal{V}$. Let $\hat{\nu}^{(j)} = \hat{\nu}(\hat{\mathcal{T}}_j)$. Since $\hat{\nu}(\mathcal{T})$ is non-increasing in $\mathcal{T}$, $\hat{\nu}_K = \hat{\nu}(\max_{j \le K} \hat{\mathcal{T}}_j)$; thus, from the proof of Theorem~\ref{main_fin}, $\mathrm{FWER}_{\hat{\nu}_K} = \Pr(\bigcap_{j \le K}\{V_{(1)}^N \ge \hat{\mathcal{T}}_j\})$ and $\{V_{(1)}^N \ge \hat{\mathcal{T}}_j\} \subseteq \{U_j < W\}$ on $\{V_{(1)}^N < \mathcal{T}_{\mathrm{max}}\}$. Hence, on this event, $\Pr(\bigcap_{j \le K}\{V_{(1)}^N \ge \hat{\mathcal{T}}_j\} \mid \mathcal{V}) \le \Pr(\bigcap_{j \le K}\{U_j < W\} \mid \mathcal{V}) = (W/\Delta\mathcal{T})^K$

By the proof of Lemma~\ref{frisk}, on $\Omega_S \cap \Omega_D$, $\hat{\nu}^{(j)} \ge \nu$ (with $\hat{\nu}^{(j)} > \nu$ if and only if $U_j < W$). Since $\hat{\nu}_K = \min_{j \le K} \hat{\nu}^{(j)}$, we get, on $\Omega_S \cap \Omega_D$, that $\hat{\nu}_K \ne \nu$ if and only if $U_j < W$ for all $j \le K$, and $\mathbb{E}[L(\hat{\nu}_K, \nu) \mid \mathcal{V}] = (W/\Delta\mathcal{T})^K$. Finally, $0 \le W \le \min(G,\Delta\mathcal{T}) \le \Delta\mathcal{T}$, with $W = \min(G,\Delta\mathcal{T})$ if $V_{(1)}^N \ge \mathcal{T}_{\text{floor}} + \Delta\mathcal{T}$.
\end{proof}

Let us denote the standard extreme value law with index $\gamma$ as $H_\gamma(x) = \exp\{-(1+\gamma x)^{-1/\gamma}\}$, $1+\gamma x > 0$, $H_0(x) = \exp(-e^{-x})$. For the points $\Pi_1 < \Pi_2 < \dots$ of a unit-rate Poisson process on $(0,\infty)$, let $D_\gamma = (\Pi_1^{-\gamma} - \Pi_2^{-\gamma})/\gamma$, with $D_0 = \ln(\Pi_2/\Pi_1)$.

\begin{theorem} \label{main_conv}
Let $a_n^{-1}(V_{(1)}^N - b_n)$ converge in distribution to $H_\gamma$, as in Assumption~\ref{noiseass}(a). For each $n$, partition $\mathcal{V}^N$ into $m_n \ge 2$ subsets with maxima $Z_1, \dots, Z_{m_n}$, let $Z_{(1)} \ge Z_{(2)}$ be the two largest of these maxima, and suppose that the point process $\{a_n^{-1}(Z_j - b_n)\}_{j=1}^{m_n}$ converges in distribution to a Poisson process.

(a) Then $a_n^{-1}(Z_{(1)} - Z_{(2)}) \Rightarrow D_\gamma$, and if $\{(a_n^{-1}(Z_{(1)} - Z_{(2)}))^K\}$ is uniformly integrable, then, for every $\Delta\mathcal{T}$,
\begin{equation} \label{eva}
\mathbb{E}\big[(\min(G,\Delta\mathcal{T})/\Delta\mathcal{T})^K\big] \le c_K(\gamma)\,(a_n/\Delta\mathcal{T})^K\,(1+o(1)).
\end{equation}

(b) If the subsets are singletons and $a_n/\Delta\mathcal{T} \to 0$, then $G/a_n \Rightarrow D_\gamma$ and
\begin{equation} \label{evb}
\liminf_{M\to\infty} (\Delta\mathcal{T}/a_n)^K\, \mathbb{E}[(\min(G,\Delta\mathcal{T})/\Delta\mathcal{T})^K] \ge c_K(\gamma).
\end{equation}
In \eqref{eva} and \eqref{evb}: $c_K(\gamma) = \mathbb{E}[D_\gamma^K] = \Gamma(2-\gamma K) \int_0^1 ((u^{-\gamma}-1)/\gamma)^K du$ for $\gamma K < 1$, with $(u^{-\gamma}-1)/\gamma = -\ln u$ for $\gamma = 0$.
\end{theorem}
\begin{proof}
(a) Since $Z_{(1)} = V_{(1)}^N$, $a_n^{-1}(Z_{(1)} - b_n) \Rightarrow H_\gamma$, thus the limit process has mean measure $\mu_\gamma(x,\infty) = -\ln H_\gamma(x) = (1+\gamma x)^{-1/\gamma}$. Using the mapping theorem, we conclude that the points of the limit process in decreasing order are distributed as $(\Pi_i^{-\gamma}-1)/\gamma$, thus $a_n^{-1}(Z_{(1)} - Z_{(2)}) \Rightarrow D_\gamma$ \citep{Resnick1987}. Next, since $(\Pi_1, \Pi_2)$ is distributed as $(\eta\,\Pi_2, \Pi_2)$ with $\eta \sim \mathrm{Unif}(0,1)$ independent of $\Pi_2 \sim \mathrm{Gamma}(2,1)$, $D_\gamma = \Pi_2^{-\gamma}(\eta^{-\gamma}-1)/\gamma$. The bound in \eqref{eva} follows from uniform integrability and $\min(G,\Delta\mathcal{T}) \le Z_{(1)} - Z_{(2)}$, since $V_{(2)}^N \ge Z_{(2)}$. (b) Here $G = Z_{(1)} - Z_{(2)}$, thus $G/a_n \Rightarrow D_\gamma$, and for any fixed $t > 0$, the left-hand side in \eqref{evb} is at least $\mathbb{E}[\min(D_\gamma, t)^K] \to c_K(\gamma)$ as $t \to \infty$.
\end{proof}

Singleton subsets correspond to i.i.d.\ noise and noise satisfying multilevel versions of $D(u_n)$ and $D'(u_n)$, e.g., stationary Gaussian noise with covariances $\rho_k \ln k \to 0$ \citep{Leadbetter1983}; then $\theta = 1$. If $\theta < 1$, subsets formed by blocks of suitably growing length $r_n = o(n)$ can be used under multilevel mixing conditions as in \citet{Hsing1988}, since block maxima are then asymptotically independent. 

\begin{theorem} \label{main_conv_na}
If the noise components are i.i.d.\ with a nondecreasing hazard rate $h$, then \[\mathbb{E}[(\min(G,\Delta\mathcal{T})/\Delta\mathcal{T})^K] \le K!\,\mathbb{E}[(\Delta\mathcal{T}\, h(V_{(2)}^N))^{-K}]\] for all $n \ge 2$ and $\Delta\mathcal{T} > 0$. For exponential noise with mean $\sigma$, this bound is $K!\,(\sigma/\Delta\mathcal{T})^K$ and asymptotically exact as $\Delta\mathcal{T}/\sigma \to \infty$.
\end{theorem}
\begin{proof}
For given $V_{(2)}^N = y$, $V_{(1)}^N$ is distributed as a noise component conditioned to exceed $y$, thus $\Pr(G > t \mid V_{(2)}^N = y) = \exp\{-\int_y^{y+t} h(s)\,ds\} \le e^{-h(y) t}$, and $\mathbb{E}[G^K \mid V_{(2)}^N] \le K!\, h(V_{(2)}^N)^{-K}$. For exponential noise, $G/\sigma \sim \mathrm{Exp}(1)$ \citep{David2003}.
\end{proof}

\textbf{Corollary 5.} Using Theorems~\ref{lemrate},~\ref{main_conv}~and~\ref{main_conv_na}, one can bound $\mathrm{FWER}_{\hat{\nu}_K}$ for arbitrary deterministic shifts and any floor detection threshold, up to $\Pr(V_{(1)}^N \ge \mathcal{T}_{\text{max}})$, which is zero for $\mathcal{T}_{\text{max}}$ of Corollary~3, and the risk up to $\Pr(\Omega_S^c \cup \Omega_D^c) + \Pr(V_{(1)}^N < \mathcal{T}_{\text{floor}} + \Delta\mathcal{T})$.

\textbf{Remark 3.} For $V_{(1)} \ge \dots \ge V_{(M)}$ with $V_{(1)} < \mathcal{T}_{\mathrm{max}}$, as $K \to \infty$, $\hat{\nu}_K$ converges almost surely to the deterministic \textit{continuous step-down thresholding} (CSDT) estimator (which for $\mathcal{T}_{\mathrm{floor}} = -\infty$ is equivalent to the spacing estimator of \citet{Passemier2012})
\[
\hat{\nu}_\infty = \min\big\{ r \in \{0, \dots, M-1\} : V_{(r+1)} - V_{(r+2)} < \Delta\mathcal{T} \text{ or } V_{(r+1)} < \mathcal{T}_{\mathrm{floor}} + \Delta\mathcal{T} \text{ or } r = M-1 \big\}.
\]
Since $(W/\Delta\mathcal{T})^K \to \mathbf{1}_{\{W = \Delta\mathcal{T}\}}$, the BCT and Theorem~\ref{lemrate} yield $\mathrm{FWER}_{\hat{\nu}_\infty} \le \Pr(W = \Delta\mathcal{T}) + \Pr(V_{(1)}^N \ge \mathcal{T}_{\mathrm{max}})$ and $|\mathbb{E}[L(\hat{\nu}_\infty, \nu)] - \Pr(W = \Delta\mathcal{T})| \le \Pr(\Omega_S^c \cup \Omega_D^c)$, where $\Pr(W = \Delta\mathcal{T}) \le \Pr(G \ge \Delta\mathcal{T})$. Thus, taking $K \to \infty$ eliminates the grid misalignment error for $G < \Delta\mathcal{T}$. However, unlike DSDT, CSDT requires computing order statistics and their spacings rather than exceedance counts, and maximizes the underestimation probability $\Pr(\hat{\nu}_K < \nu)$ on $\Omega_S^c \cup \Omega_D^c$.

\section{Extensions and Generalizations}
1. The constant step $\Delta\mathcal{T}$ can be replaced by a sequence $\Delta\mathcal{T}_M > 0$. If $a_n = o(\Delta\mathcal{T}_M)$, or $V_{(1)}^N - V_{\mathcal{C}_1}^{(2)} = o_{\Pr}(\Delta\mathcal{T}_M)$ under Assumption \ref{noiseass}(a)(ii), then $G/\Delta\mathcal{T}_M \xrightarrow{\Pr} 0$, since $G = O_{\Pr}(a_n)$ by the proof of Lemma \ref{suff}; for $a_n \not\to 0$, this requires $\Delta\mathcal{T}_M \to \infty$. Thus, Lemma \ref{grid_misalignment} and the BCT yield $\Pr(\mathcal{E}_{\text{mis}}) \to 0$, bounding $\mathbb{E}[W]/\Delta\mathcal{T}_M$ in Lemma \ref{frisk} and Theorem \ref{main_fin}. This preserves Theorem \ref{main_inf}(b), and Theorem \ref{main_inf}(a) if $\Delta\mathcal{T}_M = o(C_S(M))$ and \eqref{condtn} holds with $\Delta\mathcal{T}_M$. 

2. If condition (i) in Assumption \ref{noiseass}(a) holds, the limitation regarding deterministic components with identical shifts discussed in Remark 2 can be overcome by generalizing the stopping rule to $\Delta\hat{\nu}_k \ge m$ for $m \ge 2$ (with $M - \nu \ge m$). Partitioning $\mathcal{V}^N$ into $m$ blocks with maxima $Z_j$ as in Lemma \ref{suff} yields $0 \le V_{(1)}^N - V_{(m)}^N \le \max_{1 \le j \le m} Z_j - \min_{1 \le j \le m} Z_j \xrightarrow{\Pr} 0$, thus Theorem \ref{main_inf} then extends to groups of at most $m-1$ deterministic components with identical shifts.

3. Since $N_{\text{false}} \le \hat{\nu}$, we have $\text{FDR}_{\text{DSDT}} \le \text{FWER}_{\text{DSDT}}$; the bounds of Theorem~\ref{main_fin} and the convergence in Theorem~\ref{main_inf}(b) directly extend to $\text{FDR}_{\text{DSDT}}$ and $\Pr(\hat{\nu}_{\text{DSDT}} > \nu)$.

4. Suppose only $\nu_0 < \nu$ components satisfy the conditions of Theorem \ref{th_ens} or Theorem \ref{main_inf}(a) with $\nu$ replaced by $\nu_0$, whereas the remaining $\nu - \nu_0$ components satisfy $\Pr(\max_{\nu_0 < i \le \nu} V_i^S < V_{(1)}^N) \to 1$. On $\{\max_{\nu_0 < i \le \nu} V_i^S < V_{(1)}^N\}$, since $\Delta\hat{\nu}_j \le 1$ for all $\mathcal{T}_j > V_{(1)}^N$ and $\{V_{(2)}^N \ge \mathcal{T}_k\}$ still implies $\Delta\hat{\nu}_k \ge 2$, the upper bound of Lemma \ref{frisk} holds with $\nu_0$, yielding $\Pr(\hat{\nu}_{\text{DSDT}} = \nu_0) \to 1$.

5. If the non-zero shifts $d_i^S$ are random variables, asymptotic exact order recovery holds provided the inequalities of Assumption \ref{defgap} hold with probability tending to one as $M \to \infty$.

6. If the noise scale $\sigma > 0$ is unknown, the algorithm can evaluate standardized statistics $\tilde{V}_i = V_i / \hat{\sigma}$, where $\hat{\sigma}$ is a robust scale estimator computed from $\mathcal{V}$. Since $\nu = o(M)$, $\hat{\sigma} \xrightarrow{\Pr} \sigma$. This convergence preserves $\tilde{G} = G/\hat{\sigma} \xrightarrow{\Pr} 0$. Generating the dither $\xi$ independently of $\mathcal{V}$ ensures Lemma \ref{lem_dithering} holds.

7. The estimator $\hat{\nu}(\mathcal{T})$ depends on the search grid only through the exceedance counts, allowing direct application to multi-dimensional grids. 

\section{Simulations}
We illustrate and validate our results via Monte Carlo simulations with i.i.d.\ noise $\varepsilon_i \sim |\mathcal{N}(0, 1)|$ (Gumbel domain with $a_n \to 0$). Unless otherwise specified, $K = 1$ and $\mathcal{T}_{\text{max}}$ is given by Corollary~3 with $\delta = 1$. 

Fig.~\ref{dith} uses $M = 1000$, $\nu = 0$, $\mathcal{T}_{\text{max}} = 20$, and $\Delta\mathcal{T} \in [0.5, 10]$. For a deterministic threshold grid ($\xi \equiv 0$), $\Pr(\mathcal{E}_{\text{mis}})$ oscillates (Theorem~\ref{detfail}), while randomized grid dithering bounds $\Pr(\mathcal{E}_{\text{mis}})$ (Lemma~\ref{grid_misalignment}).

Fig.~\ref{rob} uses $M \in [10^2, 10^{4.5}]$, $\nu = 2$, $\Delta\mathcal{T} = 2$, $\mathcal{T}_{\text{floor}} = 0.7 \sqrt{2 \ln(2M)}$, $d_2 = \sqrt{2 \ln(2M)} + \Delta\mathcal{T} + 0.5$, and $d_1 = d_2 + 1.5\Delta\mathcal{T}$. As $\Pr(V_{(1)}^N \ge \mathcal{T}_{\text{floor}}) \to 1$, single-threshold detection at $\mathcal{T}_{\text{floor}}$ yields false alarms, and its risk tends to $1$. DSDT stops via the $\Delta\hat{\nu}_k \ge 2$ rule, and its risk decreases with $M$ (Lemmas~\ref{grid_misalignment}--\ref{frisk}, Theorem~\ref{main_inf}(b)). Note that, a residual risk remains constant, since $d_1 - d_2$ is fixed, while Assumption~\ref{defgap} requires growing $d_1 - d_2$.

For Fig.~\ref{comp}, we use $M = 500$, $(d_i) = (18, 14, 10)$, $\Delta\mathcal{T} = 2.5$, and $\mathcal{T}_{\text{max}} = 50$. In Panel (a), DSDT has no floor detection threshold (Corollary~1), and the data are rounded to multiples of $q \in \{0.001, 0.05, 0.1, 0.15, 0.2, 0.3, 0.4\}$. We compare DSDT with the sequential (Algorithm~1, $S = 120$ spacings) and the one-step (Condition~1, $S = 40$) tests of \cite{Burridge2006} with their fixed weights, the nominal level $7.5\%$, and $\nu_{\max} = \nu = 3$, which prevents overestimating. Quantization introduces ties \citep[Section~2.4]{David2003}: most spacings become zero, and the others become multiples of $q$. Divided by small weights, nonzero spacings exceed the standardized spacing between the weakest component and the noise maximum, thus Algorithm~1 stops prematurely and fails for $q \ge 0.15$. Since the one-step test uses fewer spacings with larger weights and DSDT uses only exceedance counts, both remain robust. The DSDT risk, mainly due to false alarms of a single dither without a detection threshold, decreases with $K$ (Fig.~\ref{comp2}). Panel (b) uses unquantized data and $\mathcal{T}_{\text{floor}} \in [1, 4]$. As in Fig.~\ref{rob}, lowering $\mathcal{T}_{\text{floor}}$ causes false alarms of single-threshold detection, and DSDT suppresses these false alarms (Theorem~\ref{main_fin}).

Fig.~\ref{comp2} uses $\nu = 2$, $\Delta\mathcal{T} = 2$, $\mathcal{T}_{\text{floor}} = -\infty$, $K \in \{1, 3\}$, and the tests of Fig.~\ref{comp} with $\nu_{\max} = 3 > \nu$. Panel (a) uses $(d_i) = (20, 10)$ and $M \in [10^{2.5}, 10^6]$, and Panel (b) uses $M = 1000$, $d_1 = d_2 + 10$, and $d_2 \in \{4, 4.5, \dots, 7, 8, 9, 10\}$. In Panel (a), the DSDT risk decreases as $M$ and $K$ increase (Theorems~\ref{main_inf}(b) and~\ref{lemrate}), whereas the risks of both tests grow with $M$. After removing the deterministic components, Algorithm~1 tests the noise at a fixed level, its FWER does not vanish and (with the fixed weights) exceeds the nominal level for large $M$. The one-step test misses the weakest component if the standardized spacing between deterministic components exceeds that between the weakest component and the noise maximum. This occurs more often for larger $M$, as the noise maximum grows, and, in Panel (b), for $d_2 \le 8$, where the spacing $d_1 - d_2 = 10$ dominates. In Panel (b), FWER is $0.36$ for single-threshold detection, and its bounds for $K = 1$, $K = 3$, and CSDT are $0.045$, $0.0036$, and $2.8 \cdot 10^{-5}$ (Theorem~\ref{lemrate}, Remark~3). Note that, if the weakest component exceeds $\max(V_{(1)}^N, \mathcal{T}_{\text{floor}})$ by $s < \Delta\mathcal{T}$, DSDT misses it with probability $1 - (s/\Delta\mathcal{T})^K$; CSDT always misses it (Remark~3). Thus, at the step $\Delta\mathcal{T}$, single-threshold detection yields the smallest risk for $d_2 = 4$, $K = 1$ for $d_2 \in \{4.5, 5\}$, $K = 3$ for $d_2 \in \{5.5, 6\}$, and CSDT for $d_2 \ge 6.5$. 

\begin{figure}[htbp]
    \centering
    \includegraphics[width=0.6\textwidth]{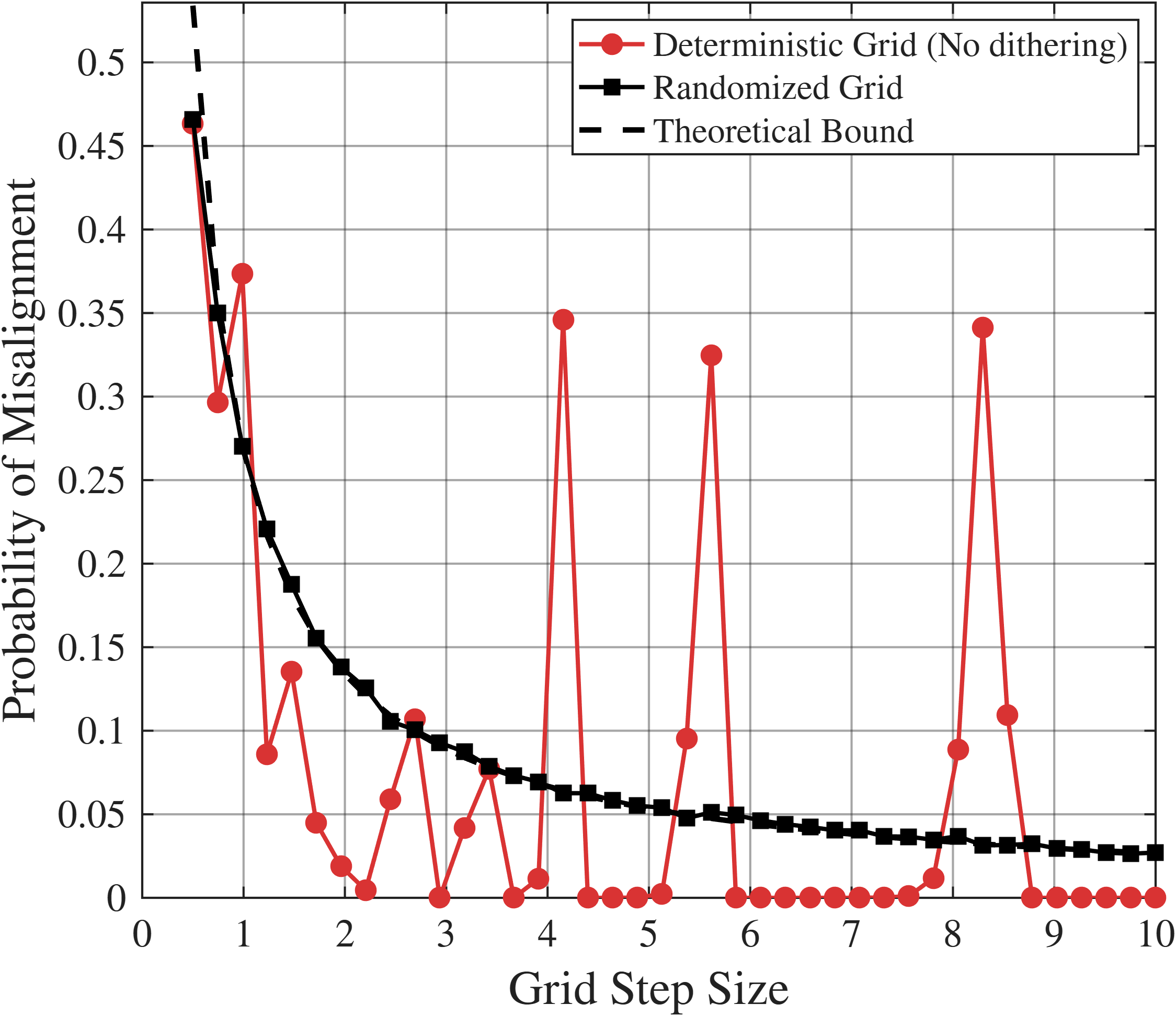}
    \caption{Probability of threshold grid misalignment $\Pr(\mathcal{E}_{\text{mis}})$ versus threshold grid step size $\Delta\mathcal{T}$.}
    \label{dith}
\end{figure}

\begin{figure}[!htb]
    \centering
    \includegraphics[width=0.6\textwidth]{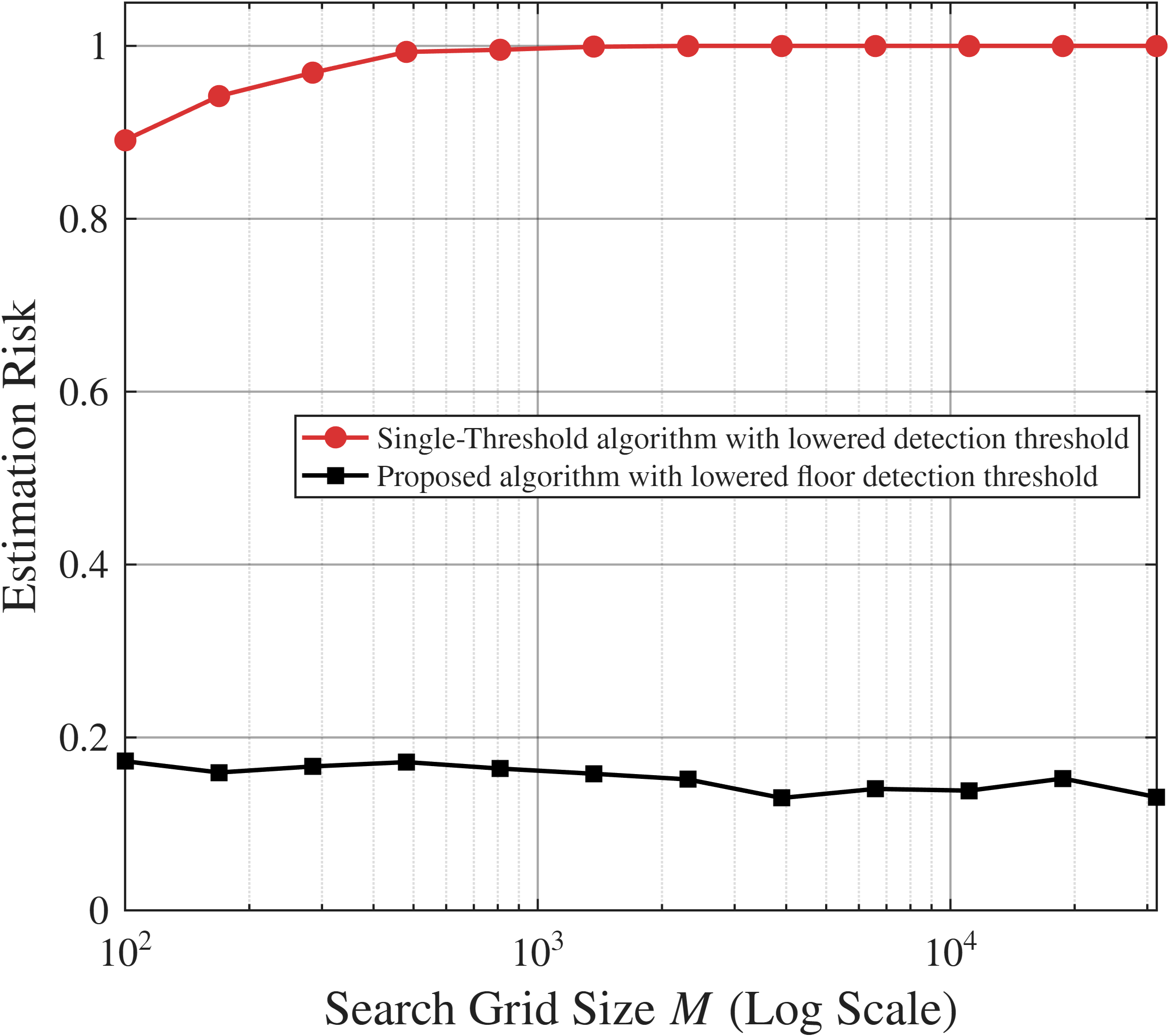}
    \caption{Estimation risk $\mathbb{E}[L(\hat{\nu}, \nu)]$ versus search grid size $M$ under a lowered floor detection threshold.}
    \label{rob}
\end{figure}

\begin{figure}[!htb]
    \centering
    \includegraphics[width=0.9\textwidth]{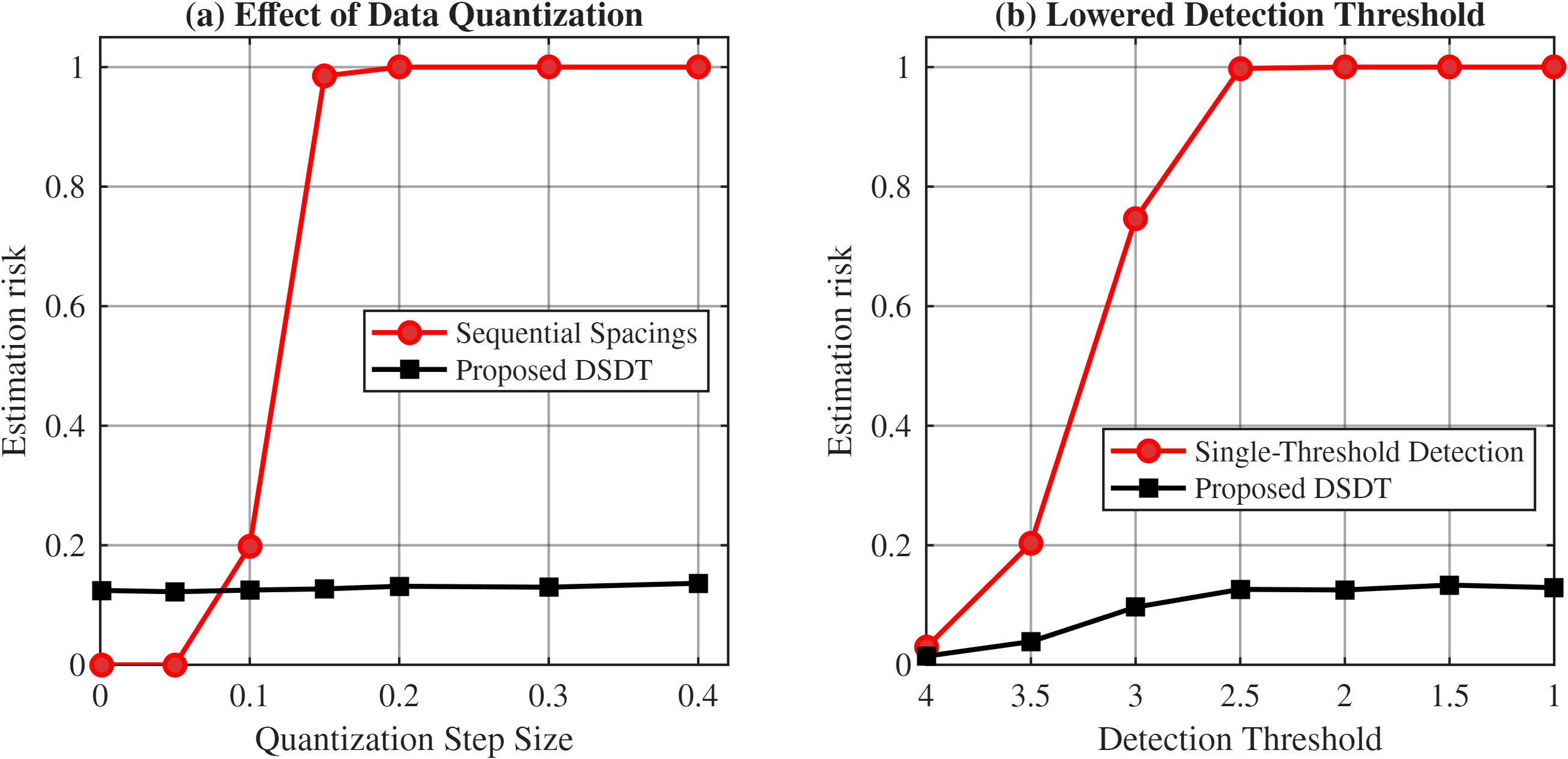}
    \caption{Estimation risk $\mathbb{E}[L(\hat{\nu}, \nu)]$ versus (a) quantization step size $q$ and (b) floor detection threshold $\mathcal{T}_{\text{floor}}$. Sequential and One-Step Spacings: tests of \citet{Burridge2006}.}
    \label{comp}
\end{figure}

\begin{figure}[!htb]
    \centering
    \includegraphics[width=0.9\textwidth]{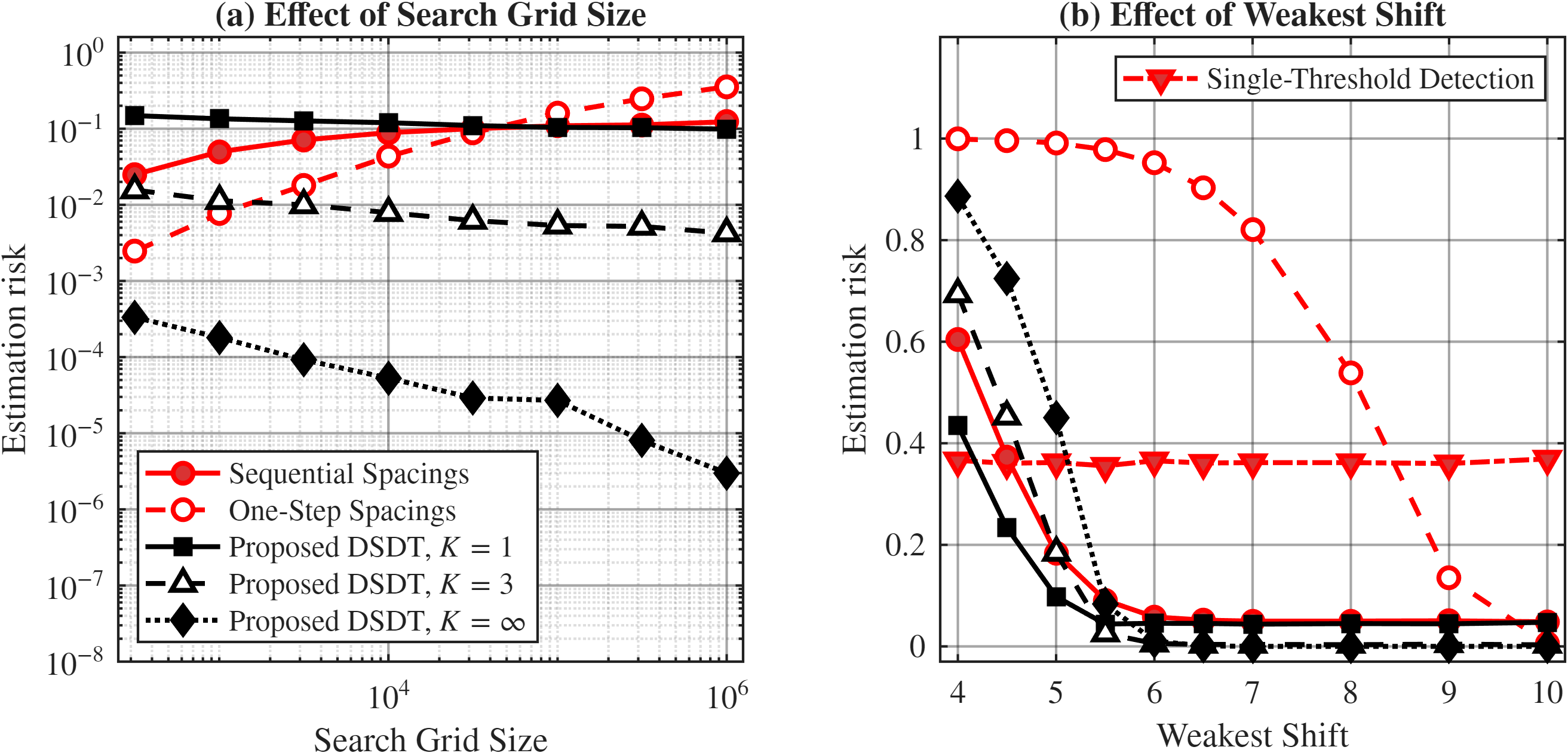}
    \caption{Estimation risk $\mathbb{E}[L(\hat{\nu}, \nu)]$ of DSDT with $K$ dithers ($K = \infty$: CSDT) and of the tests of \citet{Burridge2006} versus (a) search grid size $M$ and (b) weakest shift $d_2$.}
    \label{comp2}
\end{figure}

\section{Conclusion}
In this paper, a differential step-down thresholding algorithm with randomized grid dithering is proposed and analyzed for robust model order selection in stationary noise. The proposed algorithm uses the asymptotic value clustering of upper noise extrema to establish a stopping rule, mitigating the false alarms inherent to single-threshold detection. The threshold grid misalignment error, induced by threshold grid evaluations, is bounded using the randomized grid dithering. Our theoretical analysis and numerical simulations demonstrate that the proposed algorithm achieves asymptotic exact order recovery under the 0-1 loss function. Moreover our analysis confirms the capability of the proposed algorithm to asymptotically prevent the inclusion of noise excursions even when the floor detection threshold is lowered below the minimax threshold. Furthermore, the obtained results ensure that DSDT is a universal false alarm suppressor for an arbitrary single-threshold algorithm. We also generalize our approach to multiple independent dithers to improve the convergence rate. For an arbitrary number of dithers, we establish asymptotic bounds for the convergence rate under extreme value laws and non-asymptotic bounds under nondecreasing hazard rates. Finally we show that the results of this paper can be extended to unseparated deterministic shifts, random shifts, unknown noise scales, and multi-dimensional search grids. Also we show that the proposed algorithm is robust to data quantization.

\section*{Acknowledgments}
Funding: This work was supported by the Russian Science Foundation, project No. 24-11-00191.

\bibliographystyle{apalike}
\bibliography{Kharinstat}

\end{document}